\documentclass[sigconf]{acmart}
\usepackage{amsmath}
\usepackage{algorithm}
\usepackage{algorithmic}
\usepackage{bbm}
\usepackage{multirow}
\usepackage{subcaption}
\usepackage{ulem}
\usepackage{xspace}
\usepackage{amsfonts}
\usepackage{amsthm}

\allowdisplaybreaks

\newcommand{\KLSHP}{KL-SHP\xspace}
\newcommand{\SHPI}{SHP-I\xspace}
\newcommand{\SHPII}{SHP-II\xspace}
\newcommand{\xhdr}[1]{\vspace{1.0mm}\noindent{{\bf #1.}}\hspace{0.5mm}}
\newcommand{\E}{\mathbb{E}}

\renewcommand\footnotetextcopyrightpermission[1]{}

\begin{document}

\title{Prioritized Restreaming Algorithms \\ for Balanced Graph Partitioning}

\author{Amel Awadelkarim}
  \affiliation{\institution{Stanford University}  \city{Stanford} \state{CA}}
  \email{ameloa@stanford.edu}
\author{Johan Ugander}
  \affiliation{\institution{Stanford University}  \city{Stanford} \state{CA}}
  \email{jugander@stanford.edu}

\begin{abstract}
Balanced graph partitioning is a critical step for many large-scale distributed computations with relational data. As graph datasets have grown in size and density, a range of highly-scalable balanced partitioning algorithms have appeared to meet varied demands across different domains. As the starting point for the present work, we observe that two recently introduced families of iterative partitioners---those based on restreaming and those based on balanced label propagation (including Facebook's Social Hash Partitioner)---can be viewed through a common modular framework of design decisions.
With the help of this modular perspective, 
we find that a key combination of design decisions
leads to a novel family of algorithms with notably better empirical performance than any existing highly-scalable algorithm on a broad range of real-world graphs. The resulting {\it prioritized restreaming algorithms} employ a constraint management strategy based on multiplicative weights, borrowed from the restreaming literature, while adopting notions of priority from balanced label propagation to optimize the ordering of the streaming process. Our experimental results consider a range of stream orders, where a dynamic ordering based on what we call {\it ambivalence} is broadly the most performative in terms of the cut quality of the resulting balanced partitions, with a static ordering based on degree being nearly as good.
\end{abstract}

\maketitle

\section{Introduction}
Graphs are ubiquitous structures in computer science for representing a host of real-world systems, including social and information networks, biological networks, and meshed domains in physics simulations. The scale of such systems of interest continue to grow, particularly in domains connected to online social data. The modern World Wide Web hosts tens of billions of webpages (nodes) with trillions of links (edges) between them. Facebook serves billions of monthly active users, plus hundreds of millions of pages, events, and groups, all interacting with each other through network structures.
Similarly, Twitter sees hundreds of millions of monthly active users interact by sharing and liking each others content. 
In all these examples, graph-wide computations---most notably in the service of ranking and recommendation problems---are central to the core functions of many products and services.

Unfortunately, 
large-scale computations are expensive; 
these graphs account for terabytes of compressed data~\cite{Stanton2012} and most computations over such datasets are intractable for a single machine to perform. The typical solution to this problem involves partitioning the input graph across a number of machines and using parallel algorithms for these computations, thereby increasing computational efficiency in terms of both network latency and runtime~\cite{Buluc13}. 

The question becomes how do we ``best'' partition the network to achieve these performance gains? The answer to this question is often highly context-specific. Indeed, some problems are best distributed by partitioning the node set, while others are best distributed by partitioning the edge set~\cite{gonzalez2012powergraph}. 
In this work we focus on applications motivated by partitioning the node set (without replication), and approaches to efficiently partitioning the node sets of large empirical graphs, a difficult task~\cite{leskovec2009community}. As further motivation for our work, balanced node set partitioning has recently been used in causal inference to improve the design of experiments in networked settings through a procedure dubbed graph cluster randomization \cite{ugander2013graph,saveski2017detecting}; this area of work specifically motivates the search for good $k$-way balanced partitions for very large $k$.

A common approach to node set partitioning is a simple hashing of the node set, effectively distributing nodes uniformly at random across clusters (or machines)~\cite{Sarwat2012, Shao2013, Malewicz2010}. But more intelligent approaches to partitioning can greatly improve the runtime of these distributed algorithms~\cite{UB2013, Stanton2012}. One of the important requirements of this partitioning task, compared to generic graph clustering tasks, is that we seek to balance the computational load associated with each cluster
of the partitioning. In this work we will focus on contexts where the computational load is constant per node, but the algorithms we consider and introduce can all be easily modified to account for non-uniform/weighted loads~\cite{NU2013}.

Enter the problem of interest, \textit{balanced graph partitioning}: given an input graph, how can we partition the node set to (1) maintain balanced loads across $k$ clusters, or \textit{shards}, while (2) minimizing some objective function. We focus on the \textit{edge-cut} objective~\cite{Buluc13}---minimizing the number of edges which span multiple shards---as it closely aligns with the literature. We recognize that minimizing edge-cut may not fully depict workload performance in practice~\cite{Pacaci2019}, but we use this objective as a proxy and to catalog the effect of various design decisions on this outcome. Other examples of objective functions include \textit{fanout} minimization for hypergraphs~\cite{Kabiljo2017} and the variance minimization in graph cluster randomization~\cite{ugander2013graph}.

Unfortunately, finding an exact solution to the edge-cut problem is infeasible for even modest graphs: when the number of shards is two, this problem equates to the minimum bisection problem, which is classically NP-hard~\cite{Andreev06} and for which there are no known efficient algorithms with good approximation guarantees. That said, there is a large body of work on practical, albeit heuristic, algorithms that perform well empirically successful across a range of relevant large-scale graph datasets. 

Recent work on scalable practical algorithms for graph partitioning has been driven largely by research at companies that manage some of the world's largest relational datasets~\cite{UB2013, Kabiljo2017, Martella2017, Stanton2012, Tsourakakis14, aydin2019distributed, bateni2017affinity, Duong2013}. 
In this work, we build a common framework around three such recent algorithms that are generally regarded as at or near the state-of-the-art for different objectives: Balanced Label Propagation (BLP)~\cite{UB2013}, Restreamed Linear Deterministic Greedy (reLDG)~\cite{NU2013}, and Social Hash partitioner (SHP)~\cite{Kabiljo2017}. 
In our experimental evaluations we also benchmark against a recent high-performance algorithm based on linear embeddings~\cite{aydin2019distributed}, an approach that is not obviously related to these other approaches.

BLP and SHP belong to a family of algorithms based on label propagation. Starting from an initial assignment they iteratively conduct node relocations to achieve higher quality partitions. ReLDG is an example of what are called restreaming algorithms~\cite{NU2013}, processing the node set serially in repeated passes, with each node placed according to an assignment rule designed to achieve balance. Streaming algorithms are commonly motivated by a highly restricted computational framework where one is attempting to make node assignments while the graph is in transit, being moved and/or loaded (during ETL, in the language of data warehousing). As such, the only \textit{stream orderings} of the node set tested prior to this work were random, breadth-first-search (BFS), and depth-first-seach (DFS) to mimic the order obtained by a web-crawler or equivalent process \cite{Stanton2012}. Our work is thus the first to (1) benchmark the latter against scalable non-streaming algorithms, and (2) explore strategic stream orderings, the order in which the node set is considered by the algorithm. We call these algorithms \textit{prioritized restreaming algorithms} for balanced graph partitioning.

\xhdr{Our contribution}
The contribution of this work can be summarized in three points:
\begin{enumerate}
\item We provide benchmarking that has been absent from the literature, showing that the existing restreaming algorithm reLDG outperforms BLP and SHP\footnote{Our implementation of SHP has been adapted to minimize the edge-cut objective, rather than fanout. See Section~\ref{methods}.} on a range of real-world graphs. 
\item We modularize the three algorithms in our discussion
and notice that they are in fact three different combinations of design decisions within a common framework in terms of how they manage constraints, node priority, and a concept we call incumbency.

\item We introduce both static and dynamic stream orderings, where the latter can vary between stream iterations, as a way to inject priority into streaming algorithms for balanced graph partitioning. In particular, one such dynamic ordering, \textit{ambivalence} ordering, produces the best or nearly best results in all test cases, followed closely by a static degree ordering.
\end{enumerate}
By illustrating how these existing algorithms can be viewed under the same framework, we highlight potential improvements in each. While not all of these directions lead to improvements (we document several failed attempts at improvement),
the (dynamic) stream ordering contribution stands out as a significant advancement of the state-of-the-art. Our results are supplemented with extensive empirical investigations of the role of various design decisions, presented in Section \ref{dissect}, in these algorithms.

\xhdr{Paper structure} Section~\ref{def} formally defines the problem of interest and sets up the notation used in the remainder of this work. In Section~\ref{methods}, the three aforementioned iterative techniques---BLP, SHP, and reLDG---are presented as they exist in the literature. We introduce a decomposition of the algorithms into their modular components in Section~\ref{dissect}, laying out the taxonomy we will refer to for the remainder of the work. In Section~\ref{orders} we discuss stream orders, and introduce a novel stream order inspired by the other non-streaming methods. Section~\ref{results} studies empirical evaluations of the algorithms on a variety of graphs. Finally, Section \ref{conc} concludes and summarizes our main findings.

\subsection{Related Work}
\label{related}
Graph partitioning and its balanced variation are well-studied problems, with major results dating back to at least 1970. Many classes of algorithms for balanced graph partitioning were omitted from this work, primarily because of their poor scaling properties when considering truly massive graphs, though we highlight some notable algorithms in this section. Borrowing nomenclature from~\cite{Buluc13}, the class of ``global'' balanced partitioners considers the entire graph in some capacity and strives to achieve a solution to adjacent problems with some version of theoretical guarantees, e.g.~spectral partitioning or max-flow/min-cut-based algorithms~\cite{Brunetta1997, Arm08} for bipartitioning. 
Given a bipartitioning algorithm, one can achieve a $k$-way partition by recursively cutting the graph $\log_2 k$ times. The earliest iterative algorithms for $k$-way partitioning were based on recursive schemes for bisection~\cite{Kernighan1970, Fiduccia1982}. However, these methods are less than ideal in our context for a few reasons: (1) spectral algorithms become impractical to compute for extremely large graphs, and in this work we focus on the frontier of truly massive graphs and (2) recursive bisection greatly restricts the $k$-way partition. Hence, we focus our work on direct $k$-way partitioning algorithms.

Another class of algorithms are ``coarsening'' or ``multi-level'' algorithms, which are comprised of coarsening, partitioning, uncoarsening, and refinement phases~\cite{Osipov10, Chevalier09, Meyer06}. These methods strive to harness the theoretical benefits of the previously mentioned techniques, but on smaller contracted graphs. METIS~\cite{Karypis1998fast, Karypis98}, a family of partitioning algorithms, is an example of a multilevel method, and currently represents the state-of-the-art in partition quality (for the edge-cut objective). As such, we present these results in our experiments in Section \ref{results}.

However, though METIS has a multi-threaded implementation \cite{Lasalle2013}, these methods generally require significant resources in terms of memory and time~\cite{Pacaci2019}, so we focus our attention on the ``local-improvement'' or ``iterative'' class of algorithms.
This class makes adjustments to feasible partitionings using only information at the local level for each node. BLP, SHP, and reLDG all fall into this class.
Other examples include the classic Kernighan--Lin~\cite{Kernighan1970} heuristic and its descendants~\cite{Fiduccia1982, UB2013, Kabiljo2017}, other streaming algorithms~\cite{Stanton2012, Stanton2014, NU2013, Tsourakakis14}, max-flow-based local improvements~\cite{sanders2011engineering}, and diffusion-based methods, which are primarily used for clustering with a few extensions to partitioning~\cite{Meyer09, Pellegrini07}. This class is attractive to researchers and engineers for their speed, ease of implementation, and relatively intuitive nature. 

\section{Problem Definition}
\label{def}
In this work we study iterative algorithms for solving the balanced $k$-way partitioning problem: given an undirected graph $G=(V,E)$ on $|V| = n$ nodes and $|E|=m$ edges, an integer $k$, and an imbalance parameter $\epsilon$, find a \textit{partitioning} $P = \{V_1,\hdots, V_k\}$ of the node set into $k$ disjoint \textit{shards} $V_i$ such that $\left\lceil (1-\epsilon) \frac{n}{k}\right\rceil\leq|V_i| \leq \left\lceil (1+\epsilon) \frac{n}{k}\right\rceil$ for all $i$, and the number of cross-shard edges is minimized. Formally, the edge-cut objective looks to minimize the size of the cut set of partition $P$, 
\[
C(P) = \{(u,v) \in E ~|~ P(u) \neq P(v)\},
\]
where $P:V\rightarrow [k]$ is the shard map, mapping nodes to their shard assignment under partition $P$. As additional notation, let $N(u)$ be the neighbor set of node $u$, $N(u) = \{v\in V ~|~ (u,v)\in E\}$. In Section \ref{results}, we report our results in terms of cut quality, or \textit{internal edge fraction}, which is defined as 1 minus the cut-fraction, $1-|C(P)|/m$. Lastly, note that while we assume that $G$ is an unweighted graph, all our techniques generalize easily to weighted graphs, where balance is defined in terms of total node weight and the objective minimizes the sum of edge weights.

\section{Three Methods}
\label{methods}
In this section, we present three existing iterative algorithms---Balanced Label Propagation (BLP), Social Hash partitioner (SHP), and Restreaming Linear Deterministic Greedy (reLDG)---as they are published in the literature. This section acts as a quick introduction to the algorithms before we dissect them further in Section \ref{dissect}. As we are more concerned with design modules than optimizing performance in this work, we push discussions of complexity and parallelization of these base methods to Appendix \ref{basemethods}.
\subsection{Balanced Label Propagation}
\label{BLP}
BLP~\cite{UB2013} takes a constrained view of the label propagation literature surrounding semi-supervised learning and community detection~\cite{Zhu2002, Raghavan2007}. The BLP algorithm makes iterative, balanced improvements to an initial feasible partitioning (labelling) of the node set until an equilibrium is achieved (or a maximum number of iterations is reached). In this work, we use the simplest initialization---random balanced assignment---for comparison with other methods, though careful initialization has been shown to achieve a better equilibrium cut, depending on both context and available metadata~\cite{UB2013}.

Each iteration proceeds as follows: for every node $u\in V$, we compute its move \textit{gain}, the maximum improvement in co-located neighbor count if unilaterally relocated, defined as
\begin{equation}
g_u = \max_{i\in[k]} |N(u)\cap V_i| - |N(u)\cap V_{P(u)}|.
\label{gain}
\end{equation}
Clearly $g_u \geq 0$ for all $u\in V$. When $g_u=0$, node $u$ is effectively ``satisfied'' and gets to keep its shard assignment, a concept we formalize in Section \ref{dissect}. Nodes with $g_u > 0$ are placed in a queue to move to their target shard in order of decreasing gain. This information is funneled into a linear program that solves a circulation problem within the iteration, determining the maximum number of top nodes to move from these queues to maximize gain while abiding by constraints on each shard size. Conducting these node relocations for all shard pairs constitutes one iteration, and the process repeats until no nodes desires to move, or a maximum number of iterations is reached.

\subsection{Social Hash Partitioner}
\label{shp}

Social Hash Partitioner (SHP)~\cite{Kabiljo2017} is a two-level framework for producing and updating partitions of graph data, developed for optimizing Facebook's SocialHash~\cite{Shalita2016} infrastructure. It was built to partition more general hypergraph data~\cite{Catalyurek1999}, minimizing an objective called \textit{fanout} (the average number of shards a hyperedge spans). The algorithm is easily ``extended'' to partitioning non-hyper graphs under the traditional edge-cut objective, though that evaluation has not been done in the literature. Its mechanism for balancing shard size is a natural $k$-way extension of one of the earliest balanced partitioning algorithms for minimizing edge-cut, the Kernighan--Lin algorithm~\cite{Kernighan1970}.

Like BLP, SHP starts from an initial partitioning of the node set and makes iterative improvements to the edge-cut objective until equilibrium or a maximum number of iterations is reached. The original implementation of SHP proposed in \cite{Kabiljo2017} operates as follows: gains, as defined in Equation \eqref{gain}, are computed for each node $u\in V$. Nodes with $g_u > 0$ are bucketed in exponentially sized bins by gain to move to their target shard, $t_u$, storing two histograms per shard pair.
For all pairs, these buckets are deterministically paired and swapped from highest to lowest gain until the last bucket, where nodes are swapped in a random order until no more swaps can increase the overall gain of relocation.

As a simplifying step, in this work we modify the algorithm to store two \textit{fully sorted} queues of nodes per shard pair. Individual nodes are then paired off and swapped deterministically in order of most gain, a strict improvement in the within-iteration objective over the more easily distributed implementation in \cite{Kabiljo2017}. As this work studies the effects of these design decisions on the objective and less about computational trade-offs for distributed implementations, this simplification allows us to study the SHP algorithm in its ``best'' form. At the same time, we acknowledge that better performance within an iteration doesn't necessarily translate to better performance in equilibrium.

As another important modification, we define gain in this case to include satisfied nodes, those with $g_u = 0$, in the move queues for their second-best shard, effectively sorting by a modified form of gain over \textit{external} shards:
\begin{equation}
g'_u = \max_{i\in [k]\setminus P(u)} |N(u) \cap V_i| - |N(u) \cap V_{P(u)}|.
\label{modgain}
\end{equation}
For later reference, we denote satisfied nodes in the move queue of their best external shard as "second-best'' nodes. Such nodes are included at their own expense to possibly allow for swaps with a net-positive global gain, a hallmark characteristic of the original 1970 Kernighan--Lin algorithm. For this reason, we will denote this clarified implementation by ``\KLSHP''.

Two simplifications of \KLSHP, denoted ``\SHPI'' and ``\SHPII'' in our work, are also implemented to study the effect of constituent design decisions. In \SHPI, we both exclude second-best nodes from relocation queues and forego the prioritized ordering, randomly pairing nodes to be swapped until one queue is empty. In \SHPII, we exclude second-best nodes but still swap in a sorted order, restricting swaps to only involve nodes with positive move gains. 
Comparing \KLSHP and \SHPII showcases the effect of locally-negative (KL-style) swaps; between \SHPII and \SHPI, that of the sorted ordering.

\subsection{Restreamed Linear Deterministic Greedy}
\label{reLDG}
Restreamed Linear Deterministic Greedy (reLDG)~\cite{NU2013} falls in a subclass of iterative algorithms known as a \textit{(re)streaming} algorithms. This class is motivated by the context of single-pass online graph loading, where a program parses through a graph file, serially reading graph data from a source to a destination cluster~\cite{Stanton2012}. The multi-pass/iterative version of this approach was proposed in~\cite{NU2013}, considering restreaming methods for partitioning as potentially competitive with offline, non-streaming algorithms.

The reLDG algorithm was derived from LDG~\cite{Stanton2012}, repeatedly streaming over the node list until a maximum number of iterations is reached~\cite{NU2013}. Specifically, reLDG does the following at each iteration: for each $u\in V$, assign $u$ to the shard which satisfies
\[
\arg\max_{i\in[k]}|V_i^{(t)} \cap N(u)|\cdot\left(1-\frac{x_i^{(t)}}{C}\right).
\]
Here $V_i^{(t)}$ holds the current population of shard $i$, from the previous or the current stream (when applicable, if the node has already been ``seen'' this iteration), $x_i^{(t)}$ holds the number of nodes assigned to $i$ in the current stream, and $C$ is the shard capacity constraint, $C = (1+\epsilon)\cdot\left \lceil \frac{|V|}{k}\right \rceil$. Notice that as the shards begin to fill up, the ``multiplicative weight'' $1-x_i^{(t)}/C$ approaches zero, eventually eliminating filled shards from consideration.

Unsurprisingly, the position of a node in the stream order plays a large role in the quality of the resulting partition around that node. Nodes at the beginning of the stream are not yet impacted by the multiplicative weight, while the assignment of nodes at the end may be dominated by this term. The previous study of LDG and reLDG focused on a random (persistent) order, with some consideration given to BFS/DFS order in the original LDG work. We revisit the idea of stream orderings in Section~\ref{orders}.

\section{Taxonomy of Balanced Partitioning Algorithms}
\label{dissect}
In this section, we introduce a decomposition of the iterative methods in Section~\ref{methods} into modular parts, developing a common taxonomy of these algorithms. The identified distinctions are (1) how node relocations are carried out, (2) whether or not nodes may be exempt from relocation due to ``incumbency'', and (3) if the algorithm makes use of ``priority''.

\xhdr{Synchronous vs.\ streaming assignment}
BLP and SHP conduct all node relocations simultaneously, utilizing information from a static snapshot of the previous partitioning. In reLDG, nodes are assigned one at a time from a serial pass over the node list, changing the assignment landscape for nodes later in the stream. In this work, we denote this distinction as \textit{synchronous} vs.\ \textit{streaming} assignment. 

\xhdr{Flow-based vs.\ pairwise constraint handling}
Between the two synchronous algorithms, BLP uses a linear program to maintain balance, maximizing relocation gain subject to constraints that the net inflow of nodes to each shard equals the net outflow, up to a desired imbalance parameter. \KLSHP on the other hand simply ensures that the same number of nodes move between shard pairs. As the former has a fluid dynamical interpretation, we call this strategy \textit{flow-based} constraint handing. The latter we call \textit{pairwise} constraint handling.

\xhdr{Incumbency preference}
\label{statusquo}
Recall that in BLP, \SHPI, and \SHPII, only nodes with gain $g_u>0$ are eligible for relocation. All other nodes are reassigned to their previous shard assignment. On the other hand, \KLSHP allows for suboptimal movement via relocating ``second-best'' nodes for a globally-positive swap. ReLDG serially assigns each node at every iteration, potentially evicting nodes late in the stream which were assigned to a desirable shard at the previous iteration. BLP and the restricted SHP algorithms therefore have \textit{incumbency preference}, always allowing nodes to keep their previous assignment.

To parameterize this preference, we introduce a threshold $c$ for an algorithm's level of ``incumbency'', defined as the allowance for nodes with $g_u\leq c$ to keep their last assignment. In other words, only nodes with $g_u > c$ are eligible for relocation. Vanilla restreaming (random stream order) corresponds to a choice of $c=-\infty$ (no incumbency), while BLP corresponds to a choice of $c=0$. Any of the algorithms can be easily modified to accommodate $c$ as an input parameter to the method, and we explore this flexibility in Section \ref{results}.

\xhdr{Priority ordering}
We define \textit{priority} in this work as an ordering of how non-incumbent nodes are considered for relocation. Both BLP and \KLSHP prioritize gain, Eq~\eqref{gain} or \eqref{modgain}, in conducting node relocations. They utilize sorting within relocation queues to move nodes with the highest gain first, thereby directly optimizing edge-cut. On the other hand, vanilla reLDG does not prioritize any metric in relocating nodes; nodes are prioritized randomly in the stream order. This fact highlights an opportunity for improvement among this family of algorithms, which we explore in the following section.

\section{Priority Through Stream Orders}
\label{orders}
We now consider how node prioritization can be incorporated into the reLDG algorithm through a thoughtful choice of the order in which nodes are streamed. An adversarial demonstration for LDG (a single stream iteration) given by Stanton \& Kliot~\cite{Stanton2012} clearly shows how the stream order of nodes can play a large role in the final cut quality of (re)LDG. That said, the original work focused only on random, BFS, and DFS stream orders within the graph loading context. In this section, we investigate alternative \textit{static} as well as \textit{dynamic} prioritized orderings. As forward pointers, we study the performance of reLDG with these orders in Section~\ref{results}, most specifically in Table~\ref{tab:results}. The rank correlation between different stream orders is inspected in Figure~\ref{fig:kendall}. We discuss the complexity of computing these stream orders in Appendix~\ref{streamcomplexity}.

\subsection{Static}
We classify all of the previously considered orders---(persistent) random, BFS, and DFS---as ``static'', as they are defined based on graph properties alone and need not be updated between iterations. Of these, we consider only random and BFS, rooted at the largest-degree node. BFS (1) broadly outperforms DFS in~\cite{Stanton2012} and (2) is a good surrogate for the node order obtained from a web crawler or similar graph exploration process. Random order is analogous to the random assignments we use to initialize BLP and the SHP algorithms. We add two additional prioritized static orderings for consideration: degree and local clustering coefficient~\cite{Watts1998}, both in decreasing order.

\subsection{Dynamic}
Dynamic stream orders are updated between iterations of a restreaming algorithm. We introduce two prioritized dynamic stream orders in this work: gain order, as defined in Equation~\eqref{gain}, and a new ordering we call \textit{ambivalence}. Note that random order could be implemented in a dynamic manner, shuffling the order between iterations. However, we choose to focus our attention on dynamic stream orders which leverage updated information in the network.

\xhdr{Gain}  A natural first idea is to follow the lead of synchronous algorithms such as BLP and stream nodes in decreasing gain order. In other words, place nodes that stand to gain the most early in the stream, and those that do not stand to gain much late in the stream. However, this ordering easily backfires in the streaming setting, where nodes with a low gain value, e.g.~$g_u=0$, may have little to gain but at the same time risk incurring a significant \textit{loss} by moving. Placing such nodes at the end of the stream makes them likely to be ``evicted'' from their satisfactory assignment.

\xhdr{Ambivalence}
To remedy the above issues with gain-sorted streaming, we propose a novel metric, {\it ambivalence}, as a prioritized stream order, streaming in increasing order. That is, nodes that strongly prefer to either move or stay in place are placed early in the order, thereby giving the nodes a good chances at getting what they want, whereas nodes which are more ``ambivalent'' are streamed later. We define the ambivalence of node $u$, $a_u$, as the (negative) maximum difference in co-assigned neighbors when contrasting the current assignment with the best possible external assignment:
\begin{equation}
a_u = - \max_{i\in[k]\setminus P(u)} \left| |N(u)\cap V_i| - |N(u) \cap V_{P(u)}|\right|.
\label{ambiv}
\end{equation}
The higher (less negative) the score $a_u$, the smaller the gap in neighbor co-location count between the node's current assignment and the best other shard. As ambivalence ranges from negative degree to 0, the order has a tendency to push low degree nodes towards the end of the stream. In Section~\ref{stream} we observe a high correlation between ambivalence and degree. Further, in Appendix~\ref{bounds} we show that the expected initial ambivalence is upper and lower bounded by monotonic (linear) functions of the node degree.

\xhdr{Initialization}
The two dynamic schemes are defined relative to a partition, $P$. Specifically, they are undefined during the first pass of reLDG. As such, we define both as using degree order for their first iteration (and do so in Section~\ref{results}); degree order gives the best empirical performance of the static orders after many iterations, as in Table~\ref{tab:results}, but also after one iteration (not shown).

\section{Results}
\label{results}
We design experiments to answer the following questions:

\begin{enumerate}
\item How do the presented algorithms for balanced graph partitioning, which previously haven't been well-benchmarked, compare in terms of cut quality?
\item What role do the modules in Section~\ref{dissect} play in the performance of these methods?
\item How does the performance of our prioritized reLDG algorithm scale with increasing $k$?
\item How does stream order affect the performance of reLDG?
\end{enumerate}

\begin{table}[t!]
\centering
\caption{Test networks, all from the SNAP repository~\cite{Leskovec2014}. Here $\bar{d}$ is average degree and LCC denotes the percent of nodes in the largest connected component.}
\begin{tabular}{|l || c | c | c | c | c | c |}
\hline
Graph 			& 	$n$		& 	$m$		&$\bar{d}$&      LCC 	& Type 	\\\hline
\texttt{pokec} 		&  1,632,803 	& 22,301,964	& 27.32	& 100\%	& Social 	 \\
\texttt{livejournal}	&   4,847,571	& 43,110,428	& 17.79	& 99.9\%	& Social 	 \\
\texttt{orkut} 		&  3,072,441	& 63,464,467	& 41.31	& 100\%	& Social 	 \\\hline
\texttt{notredame} 	&   325,729	& 1,103,835	& 6.78	& 100\%	& Web 	 \\
\texttt{stanford} 		&   281,903	& 1,992,636	& 14.14	& 91\%	& Web 	 \\
\texttt{google} 		&   875,713	& 4,322,051	& 9.87	& 98\%	& Web 	 \\
\texttt{berkstan} 	&   685,230	& 7,600,595	& 19.41	& 96\%	& Web 	 \\
\hline
\end{tabular}
\label{tab:networks}
\end{table}

\begin{table*}[t!]
\centering
\caption{Internal edge fractions of 16-shard partitioning after 10 iterations of each method under exact balance ($\epsilon=0$). Highest quality, excluding METIS (0.001), in bold. As a family, reLDG and its various stream orderings outperform the top performer of the synchronous class, with the best performance coming from ambivalence (4 of 7 networks). Of the synchronous methods, \SHPI and \SHPII show superior results over their more advanced counterparts on all graphs.}
\renewcommand{\arraystretch}{1.2}
\begin{tabular}{| l || c | c | c | c || c | c | c | c || c | c || c |}
\hline
\multirow{2}{1cm}{} & \multicolumn{4}{c||}{Synchronous} & \multicolumn{6}{c||}{Streaming (reLDG)} & \\
\cline{2-12}
Graph    			& \SHPI 	&\SHPII 	& \KLSHP & BLP	& Random&  CC	&BFS		&    Degree	& Ambivalence	&  Gain 	& METIS\\
\hline
\texttt{pokec}         	&  0.578   	&  0.595	&  0.585	& 0.532 	& 0.675	&0.681	&  0.698		& \textbf{0.716}	&  0.712		&   0.618 	& 0.827\\
\texttt{livejournal}	&  0.626   	&  0.648	&  0.625	& 0.617	&   0.674	& 0.666 	&  0.731		& 0.745   		& \textbf{0.749}	&   0.671	& 0.899\\
\texttt{orkut}         	&  0.535  	&  0.555	&  0.534  	& 0.531 	&   0.650 	& 0.628 	&  0.665  		& \textbf{0.689} & 	0.679	&   0.626	& 0.711\\\hline
\texttt{notredame}     	&  0.783 	&  0.635	&  0.652 	& 0.612 	&   0.882	& 0.864 	& \textbf{0.929}	&  0.902		& 	0.924	&   0.878	& 0.982\\
\texttt{stanford}       	&  0.737  	&  0.711	&  0.697 	& 0.629 	&   0.856	& 0.844 	&  0.891		&  0.900   		& \textbf{0.916}	&   0.793 	& 0.973\\
\texttt{google}         	&  0.670  	&  0.603	&  0.616 	& 0.606 	&   0.848 	& 0.814 	&  0.868		&  0.959		& \textbf{0.964}	&   0.799 	& 0.989\\
\texttt{berkstan}       	&  0.701	&  0.652	&  0.658 	& 0.585 	&   0.858 	& 0.805 	&  0.895		& 	0.913	& \textbf{0.918}	&   0.766 	& 0.988\\\hline
\end{tabular}
\label{tab:results}
\end{table*}

\begin{figure*}[h!]
    \centering
    \includegraphics[width=\textwidth]{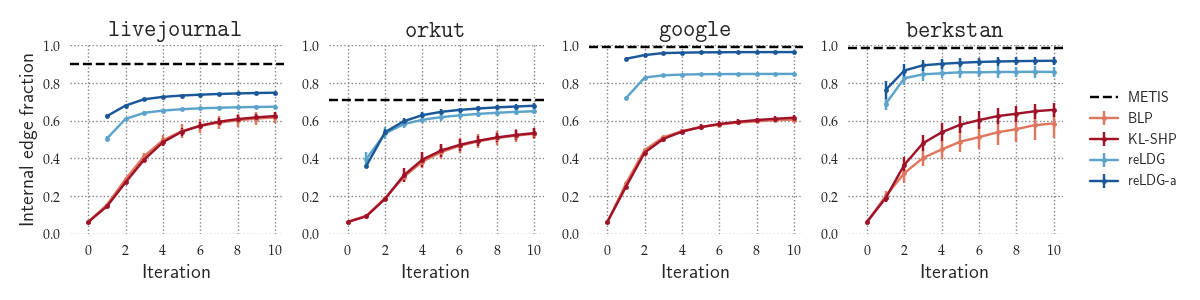}
    \caption{Internal edge fractions of $16$-shard partitioning as a function of iteration for BLP, \KLSHP, reLDG, and our best athlete, ambivalence-sorted reLDG (reLDG-a). Dotted line represents METIS (0.001). Recall that BLP and \KLSHP begin from feasible initial partitionings, hence the iteration-0 values. reLDG-a outperforms all base methods. As a family, reLDG also finds quality partitions in fewer iterations than the synchronous algorithms.}
    \label{fig:fracs}
\end{figure*}

We focus our tests of balanced partitioning algorithms on a fixed number of shards ($k=16$) and number of iterations ($t=10$), studying varied social and web networks described in Table~\ref{tab:networks}. Directed graphs were made undirected by reciprocating all edges, storing both forward and backward directed edges. Some plots focus only on the \texttt{pokec} and \texttt{notredame} graphs but are then representative of social and web graphs, respectively.  All methods are presented under exact balance, $\epsilon=0$ in the problem formulation in Section~\ref{def}. Relative performance does not change when allowing slight imbalance ($\epsilon=0.05$), so we omit imbalanced results. Given that all methods are to some extent random, if only in the handling of tie-breaks, all tabulated results were averaged over ten trials.

\subsection{Performance of the methods}
\label{overallresults}
To study question (1), we report the partition qualities of all methods---BLP, \KLSHP and its restricted forms (\SHPI, \SHPII), and reLDG with six stream orders (random, local clustering coefficient, BFS, degree, gain, ambivalence)---on all networks in Table \ref{tab:results}. In Figure \ref{fig:fracs} we further plot the internal edge fraction as a function of iteration for the existing methods (BLP, \KLSHP, and reLDG) as well as our best new method based on dynamic stream ordering, the ambivalence-sorted reLDG algorithm. 

Interpreting Table~\ref{tab:results}, reLDG with a random stream order outperforms the synchronous methods in all networks by a sizable margin, a surprising result considering that reLDG is generally regarded as further constrained by its online design. Furthermore, ambivalence order results in the best partition on four out of seven graphs (and is competitive with the best results on all seven). Several details of the relative performance of these algorithms deserve further analysis and commentary in the following sections. To answer question (2), the role of modules, we will now sequentially interpret Table~\ref{tab:results} in terms of the constraint handling and incumbency. 

\xhdr{Flow-based vs.\ pairwise constraint handling}
Analyzing the results for the synchronous methods, natural intuitions would suggest that BLP would outperform the SHP-based algorithms, as the flow-based constraint handling expands the space of allowable relocations (vs.\ pairwise handling for SHP-based methods). However, not only does \KLSHP outperform BLP on all networks, but \SHPI and \SHPII (the restricted forms of \KLSHP) give even higher quality partitions, with \SHPII performing best on the social networks and \SHPI on web (among synchronous methods). 

Recall the differences between these algorithms (see also Sections~\ref{methods} and~\ref{dissect}): the algorithms are ordered least to most advanced from left to right in Table~\ref{tab:results}. The conclusion to draw from the synchronous results is that ``less is more'', on both web and social graphs; BLP has tremendous freedom to make flow-based reassignment, and \KLSHP has strictly increased the pool of nodes eligible to move from that of \SHPII and \SHPI. Both of these design decisions, though theoretical improvements within a local iteration, perform worse once iterated for the networks in this work\footnote{Our analysis is specific to the edge-cut objective on graphs. The utility of different modules may be very different for hypergraph partitioning under the fanout objective.}. 

\xhdr{Incumbency}
From the discussion in Section~\ref{statusquo}, one of the differentiating factors between these algorithms is the different approaches to incumbency as a modular design decision. We now analyze variations on BLP, \KLSHP, and vanilla reLDG all adapted to feature a common incumbency threshold parameter, $c$, and investigate the effect of varying this threshold on the quality of the resulting partitioning. Results form varying $c$ are given in Figure~\ref{fig:c}. 

As $c$ becomes more positive, i.e., node relocation is limited to high-gain nodes, partition quality falls. Meanwhile when $c$ is negative, allowing for Kernighan--Lin-like improvements, we don't observe much change in resulting cut quality for any of the three methods. On the \texttt{notredame} graph we see a small opportunity for superior performance around $c=0$ for \KLSHP, which equates to our \SHPII algorithm, and $c=1$ for BLP. 

Upon inspection, the spike at $c=1$ for BLP on \texttt{notredame} is overwhelmingly driven by the behavior of a many degree-1 nodes. Excluding these nodes from relocation by restricting node relocation to those with $g_u > 1$ appears to helps settle the chaos of the algorithm, a result analogous to performance gains observed from strategic edge sparsification for graph clustering~\cite{Satuluri2011}. The assignment of degree-1 nodes could potentially be deferred until after other nodes are stably assigned. We do not further explore the idea of incorporating this deferral into e.g.~reLDG, an altogether different modification than setting $c=1$, but flag it for consideration by practitioners.

\begin{figure}[t!]
\centering
\includegraphics[width=\columnwidth]{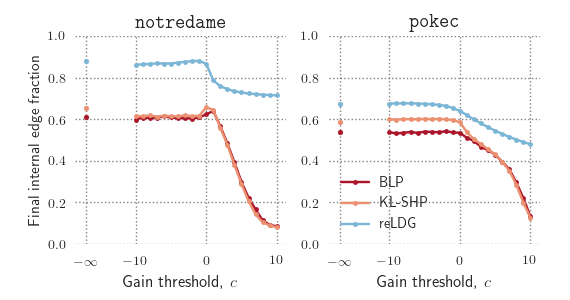}
\caption{Internal edge fractions for BLP, \KLSHP, and reLDG after 10 iterations when only nodes with gain $g_u > c$ are eligible for relocation. Across the board, increasing $c$ diminishes performance. ReLDG does best with $c=-\infty$ (see also Table~\ref{tab:results}). There may be opportunities for slightly improved performance by tuning the $c$ parameter in the cases of BLP and \KLSHP.}
\label{fig:c}
\end{figure}

\xhdr{Periodicity} Recall that synchronous algorithms compute node gains and make relocations synchronously based on a snapshot of the graph. These reassignments can cause neighboring nodes to ``pass'' each other in the move, but the implications of this effect have not been well-studied for the iterative algorithms we consider. We define the \textit{periodicity} of a node assignment at iteration $t$ as the number of iterations since the last assignment to its current shard. This quantity allows us to better understand how nodes bounce back and forth between assignments under each method. 

Formally, for node $u$ at iteration $t$, a period is defined as the minimum integer $x$ such that $P^{(t)}(u) = P^{(t-x)}(u)$, where $P^{(i)}(u)$ denotes the assignment of node $u$ at iteration $i$. In Figure~\ref{fig:periodicities} we explore the periodicity of node relocations across BLP, \KLSHP, and vanilla reLDG. In both BLP and \KLSHP, a large portion of nodes experience a periodicity of two in early iterations, especially on web graphs (\texttt{notredame} is representative). In other words, nodes are found oscillating between shards, repeatedly missing their neighbors in the move. This pathology is not present reLDG, as each node is relocated one at a time, providing full and updated picture of the assignment for each node. Periodicity helps illustrate this relative shortcoming of the synchronous algorithms compared to streaming.

\begin{figure}[t!]
  \includegraphics[width=0.5\textwidth]{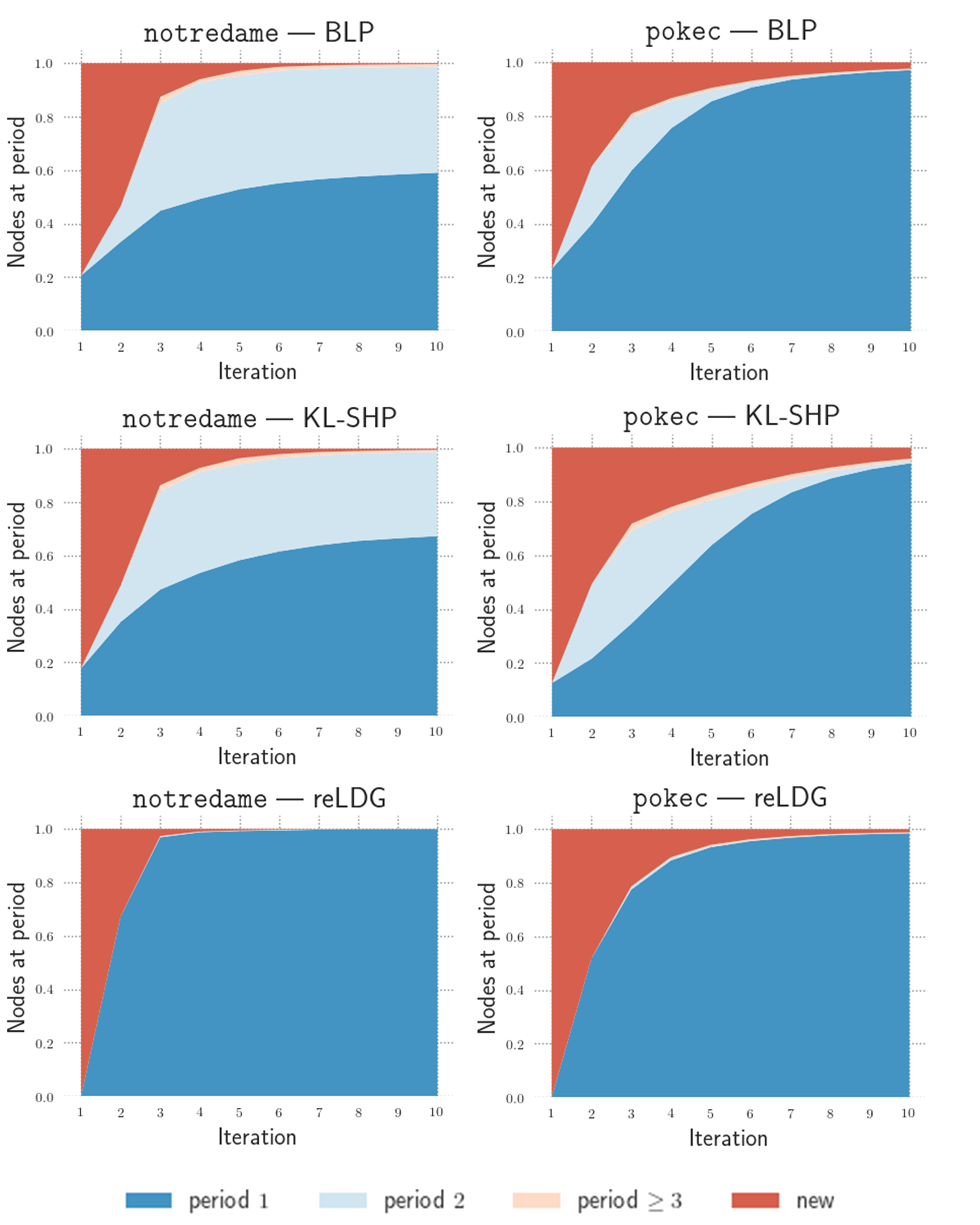}
  \caption{Fraction of nodes at a given period per iteration of BLP, \KLSHP, and reLDG. The ``new'' category denotes nodes being assigned to its current shard for the first time. On the web graph, BLP and \KLSHP have many nodes stuck at a period of 2, bouncing back and forth between two assignments.}
  \label{fig:periodicities}
\end{figure}

\begin{table}[t!]
\centering
\caption{Results from varying the number of shards, $k$. All results on LiveJournal network with $\epsilon$ given under each method name. Bold denotes most performant method (excluding METIS). 
Ambivalence-sorted reLDG (reLDG-a) consistently yields a higher quality partition than these previously benchmarked methods~\cite{aydin2019distributed}.}
\begin{tabular}{| l || c | c | c | c | c || c |}
\hline
	      		& Spinner 	&   LE/A	&   LE/C 	&   BLP	&  reLDG-a 	& METIS\\
$k$	      		&  0.05	&   0.0	&   0.0	&   0.05	&      0.0 		& 0.001\\
\hline
20          		&  0.62   	& 0.643	& 0.725	&   0.600  	&\textbf{0.733}	& 0.890\\
40    			&  0.60   	& 0.592	& 0.663	&   0.562	&\textbf{0.691}	& 0.869\\
60          		&  0.57  	& 0.570	&  0.634	&  0.537	&\textbf{0.661}	& 0.857\\
80      		&  0.56 	& 0.567	& 0.614    	&   0.520	&\textbf{0.648}	& 0.845\\
100       		&  0.54  	& 0.550	& 0.585    	&   0.517	&\textbf{0.636}	& 0.839\\
\hline
\end{tabular}
\label{tab:le}
\end{table}

\begin{figure*}[t!]
    \centering
    \begin{subfigure}[h]{0.4\textwidth}
        \centering
        \includegraphics[width=\textwidth]{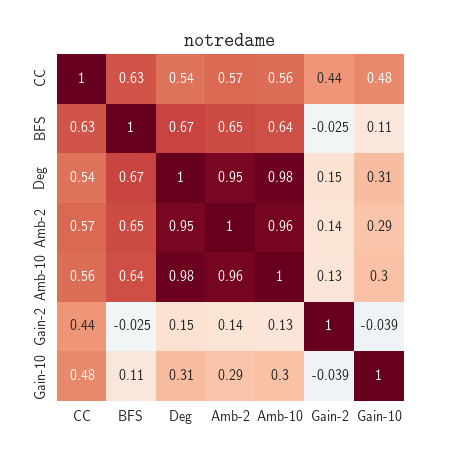}
    \end{subfigure}%
    ~
    \begin{subfigure}[h]{0.4\textwidth}
        \centering
        \includegraphics[width=\textwidth]{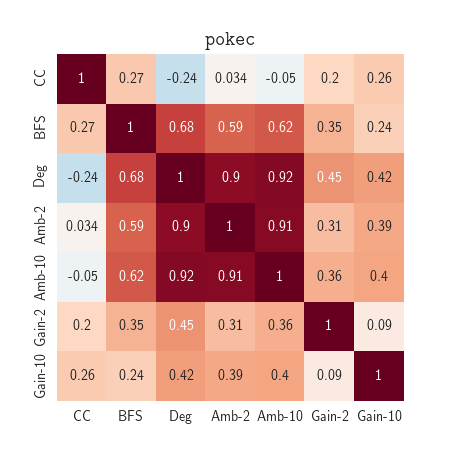}
    \end{subfigure}
    \caption{Weighted Kendall's $\tau$ rank correlation between stream orders for the \texttt{notredame} and \texttt{pokec} networks. The dynamic methods are initialized by one stream of degree order. Ambivalence and gain correlations are reported both for the 2nd and 10th iteration of the method. Degree and ambivalence are strongly correlated on both, while gain displays low correlations with the rest.}
    \label{fig:kendall}
\end{figure*}

\xhdr{Number of shards} In studying question (3), the effect of increasing $k$, we take the opportunity to benchmark our results against published numbers for other algorithms. We recreate the results in the linear embedding paper \cite{aydin2019distributed} in Table \ref{tab:le} for $k \in \{ 20, 40, 60, 80, 100\}$, borrowing values for Spinner \cite{Martella2017} and the linear embedding variations (Affinity, ``LE/A'', and Combination ``LE/C'') from that work. Their affinity mapping algorithm (LE/A) is the proposed linear embedding method, while Combination (LE/C) adds post-processing steps to optimize cuts for partitioning. Finally, we re-evaluate BLP, run with $\epsilon=0.05$ to follow results in \cite{aydin2019distributed}, and add columns for METIS and our best algorithm, ambivalence-sorted reLDG. We see that the prioritized streaming algorithm outperforms all previously benchmarked methods, and by increasing margins with increasing $k$. An explanation for this result is that while the Combination algorithm handles the embedding and partitioning steps separately, reLDG optimizes both simultaneously as a common objective.

\subsection{Performance of prioritized streaming}
\label{stream}
To address question (4), the final partition qualities under random, local clustering coefficient, BFS, degree, ambivalence, and gain orders are reported in Table \ref{tab:results} under the streaming header. First, as stated in Section~\ref{overallresults}, we see that the worst performing stream orders for reLDG result in partitions that are typically better than the most performant of the synchronous algorithms. The explanation for this surprising result is that relocating all nodes simultaneously results in neighbors ``passing'' each other, whereas all nodes see the updated partition landscape before being assigned in the streaming context. 

Unsurprisingly, gain order performs worst of all the orderings of reLDG, consistent with the intuitions discussed in Section~\ref{orders}.
Ambivalence ordering was specifically designed to solve these problems, taking into account the gain and/or loss of nodes should they move or be forced to move. Among the other orderings, ambivalence and degree are equally successful orders on social networks, and ambivalence is most performant on web graphs overall.

To understand the stream orders further, Figure~\ref{fig:kendall} shows the weighted Kendall's $\tau$~\cite{Vigna2014} correlation values for each pair of stream orders for \texttt{notredame} and \texttt{pokec} as test graphs. 
For the dynamic orders, we study both initial orders at iteration 2 (Amb-2, Gain-2) and iteration 10 (Amb-10, Gain-10). First notice that the two gain orders are very weakly correlated with each other and the other orders. This result is yet another example of the improper fit of gain-based priority in the streaming context. Next, degree and ambivalence are highly correlated measures on both test graphs---all pairs surpassing the 0.9 threshold, signaling nearly equivalent rankings~\cite{Voorhees2002}. We provide a theoretical explanation for this relationship in Appendix~\ref{bounds}, showing that the expected initial ambivalence is upper and lower bounded by monotonic (linear) functions of degree, consistent with the high correlation between the two in Figure~\ref{fig:kendall}. The fact that increasing ambivalence order is well approximated by decreasing degree, a static ordering, provides a simple alternative to ambivalence if the additional complexity of implementing dynamic stream orders is onerous.

As a final observation, on \texttt{notredame} the local clustering coefficient has moderate correlation with the other orders, whereas its correlations are more neutral, or even negative, on the \texttt{pokec} social graph. It is a well documented fact that local clustering coefficient is inversely related to degree on many complex networks~\cite{Leskovec2008, Ugander2011, Regan2003}, which would suggest the negative correlations seen in the social network.
In the \texttt{notredame} network, on the other hand, 88\% of nodes have degree $<10$, and 49\% have degree 1. Nodes with degree 1 have a clustering coefficient of 0; half of the node set is thus tied for last place in both orderings. Furthermore, when restricted to nodes with degree $<10$ we found that clustering coefficient increased with degree in the network. Thus, the strange correlations are an artifact of the degree distribution of this specific network. 

\section{Conclusion}
\label{conc}
In this work, we dissect the design decisions involved in recent highly-scalable iterative algorithms for balanced partitioning. Based on this dissection, we introduce a new class, \textit{prioritized streaming algorithms}, that leverages prioritization ideas from synchronous algorithms within the streaming setting. We contribute a novel priority ordering, ambivalence order, for streaming algorithms. When tested on various social and web graphs, we find that streaming algorithms do not suffer from observed pathologies of the synchronous assignment process used by BLP or SHP-based algorithms---namely moving or swapping neighboring nodes away from or past each other. Even vanilla reLDG (random stream order) results in higher quality partitions on all tested graphs than BLP and \KLSHP. 

The best restreaming results come from ambivalence and degree orderings, being superior on six of the seven tested graphs. Ambivalence and degree are highly correlated orderings, offering degree order as the preferred static ordering if computing ambivalence is burdensome. Though initially proposed in the online setting---moving graphs between clusters---our results clarify that restreaming algorithms are major contenders as highly scalable offline partitioners.

\xhdr{Reproducibility}
Implementations of BLP, SHP variations, reLDG, prioritized reLDG, as well as notebooks replicating plots in this paper are available at: \url{https://github.com/ameloa/streamorder}.

\begin{acks}
We thank Brian Karrer, Joel Nishimura, Arjun Seshadri, and Hao Yin for helpful comments and discussions.
This work is funded in part by a Young Investigator Award from the Army Research Office (JU, 73348-NS-YIP)
and a National Science Foundation Graduate Research Fellowship (AA, 2017237604).
\end{acks}

\bibliographystyle{ACM-Reference-Format}
\bibliography{partitioning}

\pagebreak
\clearpage
\appendix
\section{Computational considerations}
\subsection{Base methods}
\label{basemethods}
\xhdr{Complexity} 
For $k$ clusters and a graph on $n$ nodes,
both BLP and \KLSHP compute node gains and targets in $O(nk)$ operations. They then sort $k(k-1)$ queues in $O(n/k^2 \log (n/k^2))$ time\footnote{Recall that the production implementation of SHP does not fully sort each move queue to alleviate this additional complexity. See~\cite{Kabiljo2017}.}. BLP additionally solves an LP with $O(k^2)$ variables and $O(g^* k^2)$ constraints, where $g^*$ is the number of unique gain levels $g_u$. For large graphs with large degrees, $g^*$ can be quite large; one example in~\cite{UB2013} solves an LP with $12,000$ variables and $600,000$ constraints. A major achievement of SHP, it can be said, was to come up with a effective LP-less variation on BLP.  Both methods conduct node relocations in $O(n)$. 

Meanwhile, the runtime of reLDG in its proposed form is simply $O(nk)$---serially accessing each node, and finding the shard which maximizes the objective for each one---making it the most lightweight algorithm in terms of computational complexity of those discussed in this work. Alternative choices of stream order may incur additional preprocessing costs, as is discussed in Section~\ref{streamcomplexity}. 

\xhdr{Parallelization}
Though reLDG is the algorithm with the lowest serial time-complexity, BLP and SHP are both easily parallelizable, whereas the streaming algorithm is more difficult to distribute, by design. Between BLP and SHP, computing gains and sorting the move queues between each shard pair are completely independent operations per node and shard pair. All node relocations can be done in a distributed manner as well, once the LP is solved in the case of BLP. ReLDG has a parallel implementation which incurs a performance penalty that can be mitigated with more iterations \cite{NU2013}. For the sake of pure algorithmic comparison in this work, we chose to not consider parallelized implementations in our analyses.

\subsection{Stream orders}
\label{streamcomplexity}
\xhdr{Complexity}
The prioritized static and dynamic stream orders proposed in this work require sorts of the entire node set, taking $O(n\log n)$ time to sort after computing the per-node quantities of interest. Of these calculations, we compute the local clustering coefficients in $O(nd_{\max}^2)$ operations, ambivalence and gain in $O(nk)$, BFS order takes $O(n+m)$, and degree takes $O(m)$. The calculations of our static orderings are one-time up-front computations, easily stored for future use; the dynamic orderings compute their respective quantities and sort at every iteration, taking $O(n(k + \log n))$ time at each step.

\section{Expected ambivalence and degree}
\label{bounds}
To begin, note that the strong correlations in Figure~\ref{fig:kendall} are between \textit{increasing} ambivalence and \textit{decreasing} degree order. As such, we will consider negated ambivalence in this section for simplicity. Further, we adjust the definition of ambivalence in Eq.~\eqref{ambiv} from the max over absolute differences to the max over squared differences,
\begin{equation}
a_u = \max_{i\in [k] \setminus P(u)} \left(|N(u) \cap V_i| - |N(u) \cap V_{P(u)}|\right)^2.
\label{adjambiv}
\end{equation}
The adjustment bears no effect on the order of the ambivalence scores and aids the following analysis.

\begin{proposition} The expected value of the initial ambivalence in Eq.~\eqref{adjambiv} is lower and upper bounded by
\begin{align*}
 \frac{2}{k} d_u \le \E\left[a_u\right] \le \frac{2(k-1)}{k} \cdot d_u,
\end{align*}
where $k$ is the number of shards and $d_u$ is the degree of node $u$.
\end{proposition}
\begin{proof}
Let $Y \in \{0,1\}^{n\times k}$ be the matrix of node assignments under partition $P$, and $A\in \{0,1\}^{n\times n}$ denote the adjacency matrix of graph $G$.
We denote shard $i$'s column of $Y$ by $Y_i$, and node $u$'s column of $A$ by $A_u$. 

Note that $|N(u) \cap V_i| = A_u^\top Y_i$, so ambivalence can be written as:
\begin{align*}
a_u &= \max_{i\in [k] \setminus P(u)} ((Y_i - Y_{P(u)})^\top A_u)^2\\
&= \max_{i\in [k] \setminus P(u)} (Y_i - Y_{P(u)})^\top A_uA_u^\top(Y_i - Y_{P(u)})\\
&= \max_{i\in [k] \setminus P(u)} x_i^\top M_u x_i,
\end{align*}
where we define $x_i = Y_i - Y_{P(u)}$, and $M_u=A_uA_u^\top$. 

\xhdr{Lower bound}
Computing the expected value, we have
\begin{align*}
\E\left[a_u\right]&=\E\left[\max_{i\in [k] \setminus P(u)} x_i^\top M_u x_i\right]\\
&\geq \max_{i\in[k]\setminus P(u)} \E\left[x_i^\top M_u x_i\right]\\
&= \max_{i\in[k]\setminus P(u)} \text{tr}(M_u\Sigma_i) + \mu_i^\top M_u \mu_i,
\end{align*}
where $\mu_i = \E[x_i] = \E[Y_i -Y_{P(u)}]$ and $\Sigma_i = \text{Cov}(x_i)$. Computing $\mu_i$ under a random partition,
\begin{align*}
\mu_i &= \E[Y_i -Y_{P(u)}] \\
&= \E[Y_i] -\E[Y_{P(u)}] \\
&=\frac{1}{k} \mathbbm{1} - \frac{1}{k} \mathbbm{1} \\
&= 0.
\end{align*}
So we have
\begin{align*}
\E\left[a_u\right]\geq\max_{i\in[k]\setminus P(u)} \text{tr}(M_u\Sigma_i).
\end{align*}
Expanding the covariance matrix $\Sigma_i$,
\begin{align*}
\Sigma_i &= \E\left[(Y_i -Y_{P(u)})(Y_i -Y_{P(u)})^\top\right] - \mu_i\mu_i^\top\\
&= \E\left[(Y_i -Y_{P(u)})(Y_i -Y_{P(u)})^\top\right].
\end{align*}
We define the random quantity $B^{(i)} = (Y_i -Y_{P(u)})(Y_i -Y_{P(u)})^\top$. Analyzing the quantities on the diagonal and off-diagonals, respectively,
\begin{equation*}
\begin{aligned}[t]
B^{(i)}_{uu} &= (Y_{ui} - Y_{uP(u)})^2\\
&= \begin{cases}
1 & \text{if } u \in V_{i} \cup V_{P(u)},\\
0 & \text{otherwise}.
\end{cases}\\
B^{(i)}_{uv} &= (Y_{ui} - Y_{uP(u)})(Y_{vi} - Y_{vP(u)})\\
&= \begin{cases}
1 & \text{if } u,v \in V_{i} \text{ or } u,v \in V_{P(u)},\\
-1 & \text{if } u \in V_{i}, v\in V_{P(u)} \text{ or } u \in V_{P(u)}, v\in V_{i},\\
0 & \text{otherwise}.
\end{cases}
\end{aligned}
\end{equation*}
Under the initial random assignment, the probability of 1 on the diagonal is $\frac{2}{k}$, for all $u\in V$. On the off-diagonal, the probability of a value being either -1 or 1 is $\frac{2}{k^2}$. Hence, 
\[
\Sigma_i = \E[B^{(i)}] = \frac{2}{k} I, \ \forall i,
\]
where $I$ is the identity matrix. Therefore,
\begin{align*}
\E[a_u] &\geq\max_{i\in[k]\setminus P(u)} \text{tr}(M_u\Sigma_i)\\
&= \frac{2}{k}\text{tr}(M_u)\\
&=\frac{2}{k} \cdot d_u,
\end{align*}
where $d_u$ is the degree of node $u$.  

\xhdr{Upper bound}
Borrowing the same notation,
\begin{align*}
\E\left[a_u\right]&=\E\left[\max_{i\in [k] \setminus P(u)} x_i^\top M_u x_i\right]\\
&\leq \E\left[\sum_{i\in [k] \setminus P(u)} x_i^\top M_u x_i\right]\\
&= \sum_{i\in [k] \setminus P(u)} \E\left[x_i^\top M_u x_i\right]\\
&= \sum_{i\in [k] \setminus P(u)} \text{tr}(M_u\Sigma_i) + \mu_i^\top M_u \mu_i\\
&= \sum_{i\in [k] \setminus P(u)} \frac{2}{k}\text{tr}(M_u)\\
&= \frac{2(k-1)}{k}\cdot d_u.
\end{align*}
\end{proof}

\end{document}